\documentclass[]{article}
\usepackage{shortcuts}

\newcommand{\odd}{\mathrm{odd}}
\newcommand{\even}{\mathrm{even}}

\DeclareMathOperator{\wt}{wt}
\numberwithin{equation}{section}
\numberwithin{thm}{section}

\title{
    The Query Complexity of Mastermind with $\ell_p$ Distances
    \footnote{This paper is the full version of \href{http://drops.dagstuhl.de/opus/volltexte/2019/11216/}{this work} that appeared in APPROX 2019.}
}
\author{
  Manuel Fernández V \\ Carnegie Mellon University \\ \texttt{manuelf@andrew.cmu.edu} \and
  David P. Woodruff \\ Carnegie Mellon University \\ \texttt{dwoodruf@cs.cmu.edu} \and
  Taisuke Yasuda \\ Carnegie Mellon University \\ \texttt{yasuda.taisuke1@gmail.com}
}

\begin{document}

\begin{titlepage}
\maketitle
\thispagestyle{empty}
\begin{abstract}
    Consider a variant of the Mastermind game in which queries are $\ell_p$ distances, rather than the usual Hamming distance. That is, a codemaker chooses a hidden vector $\bfy\in\{-k,-k+1,\dots,k-1,k\}^n$ and answers to queries of the form $\norm*{\bfy-\bfx}_p$ where $\bfx\in\{-k,-k+1,\dots,k-1,k\}^n$. The goal is to minimize the number of queries made in order to correctly guess $\bfy$.

    Motivated by this question, in this work, we develop a nonadaptive polynomial time algorithm that works for a natural class of separable distance measures, i.e.\ coordinate-wise sums of functions of the absolute value. This in particular includes distances such as the smooth max (LogSumExp) as well as many widely-studied $M$-estimator losses, such as $\ell_p$ norms, the $\ell_1$-$\ell_2$ loss, the Huber loss, and the Fair estimator loss. When we apply this result to $\ell_p$ queries, we obtain an upper bound of $O\parens*{\min\braces*{n,\frac{n\log k}{\log n}}}$ queries for any real $1\leq p<\infty$. We also show matching lower bounds up to constant factors for the $\ell_p$ problem, even for adaptive algorithms for the approximation version of the problem, in which the problem is to output $\bfy'$ such that $\norm*{\bfy'-\bfy}_p\leq R$ for any $R\leq k^{1-\eps}n^{1/p}$ for constant $\eps>0$. Thus, essentially any approximation of this problem is as hard as finding the hidden vector exactly, up to constant factors. Finally, we show that for the noisy version of the problem, i.e.\ the setting when the codemaker answers queries with any $q = (1\pm\eps)\norm*{\bfy-\bfx}_p$, there is no query efficient algorithm.
\end{abstract}
\end{titlepage}

\newpage

\section{Introduction}
\emph{Mastermind} is a game played between two players, the \emph{codemaker} and the \emph{codebreaker}. In the 1970 original $4$-position $6$-color version of the game, the codemaker chooses $4$ colored pegs, each taking one of $6$ colors, and the codebreaker tries to guess the codemaker's $4$ pegs by making queries to the codemaker by taking a guess at the sequence of the codemaker's $4$ colored pegs. These guesses are answered by two numbers, the number of pegs guessed that are in the right position and the right color, indicated by black pegs, and the additional number of pegs of the right color but in the wrong position, indicated by white pegs.

Ever since, this game and its generalizations and variants have been studied by many computer scientists. The original version was completely characterized by \cite{knuth1977computer}, who showed upper and lower bounds of $5$ queries for deterministic strategies. The $n$-position $k$-color generalization of the game was studied in \cite{chvatal1983mastermind}, which sparked a line of research that lead to progressive improvement in upper and lower bounds for this problem, both in the original version of the game as well as in related variants of the game \cite{DBLP:journals/dm/BergerCS18}. As these variants are not the focus of this work, we refer the reader to the expositions of \cite{doerr2016playing, DBLP:journals/dm/BergerCS18} for more details on this literature.

Note that in the variant that the codebreaker only receives the black peg answers, the problem can be phrased as guessing a hidden vector based on Hamming distance queries. One can then consider many variants of the Mastermind game in which the codebreaker guesses the codemaker's hidden vector based on other distance queries. For instance, motivated by the theory of black-box complexity, \cite{afshani2019query} recently studied the variant where the distance is the length of the longest common prefix with respect to an unknown permutation. In recreational mathematics, the $\ell_1$ distance case has been studied under the name of ``digit-distance'' \cite{ginat2002digit}. When the distance between the vectors is a graph distance of a graph $G$ and we only allow for nonadaptive queries, that is when the queries cannot depend on the results of previous queries, then the query complexity is known as the \emph{metric dimension of $G$} \cite{Rodriguez-VelazquezYKO14, JiangP19}.

Another natural variant to consider is the case of $\ell_p$ distance queries. That is, the codemaker chooses a hidden vector $\bfy\in\{-k,-k+1,\dots,k-1,k\}^n$ and answers to queries of the form $\norm*{\bfy-\bfx}_p$ where $\bfx\in\{-k,-k+1,\dots,k-1,k\}^n$. This is the question we focus on in this work. We study the asymptotics with respect to $n$ and $k$, but view $p$ as a fixed constant.

\subsection{Previous work}
The above problem has been solved, even up to constant factors in the dominant term of the asymptotics for integer $p$ and $k = o(n)$. For $p=1$, note that the absolute value distance on a single coordinate is exactly the graph distance on the path graph, and the $\ell_1$ distance is exactly the graph distance on the $n$th Cartesian power of the path graph. Thus, the nonadaptive query complexity to $(2+O(\log\log n/\log n))n\log(2k+1)/\log n$ by Theorems 1 and 4 of \cite{JiangP19}. Furthermore, \cite{JiangP19} extend their techniques to a very general class of integer-valued distances in their Theorem 7, which includes $\ell_p$ distances for integer $p$. This settles the nonadaptive query complexity for $\ell_p$ distances for any fixed integer $p$ to $(2+O(\log\log n/\log n))n\log(2k+1)/\log n$ as well. Furthermore, their algorithms are efficient.

\subsection{Our contributions}
On the algorithmic side, we present Theorem \ref{thm:exact-recovery-coordinate-wise}, in which we develop a very general nonadaptive algorithm that works for any separable distance measure, i.e.\ the distance between $\bfx,\bfy\in\mathbb R^n$ is given by $f(\bfx-\bfy)$ where $f(\bfx) = \sum_{i=1}^n g_i(\abs*{x_i})$, with a mild technical assumption. This class in particular includes the smooth max ($g_i(x) = \exp(\abs{x})$) as well as many widely-studied $M$-estimator losses, such as $\ell_p$ norms for even $p\in(0,1)$, the $\ell_1$-$\ell_2$ loss ($g_i(x) = 2(\sqrt{1+\abs{x}^2/2}-1)$), the Huber loss ($g_i(x) = \abs{x}^2/2\tau$ for $\abs{x}\leq\tau$ and $g_i(x) = \abs{x}-\tau/2$ otherwise), and the Fair estimator loss ($g_i(x) = c^2(\abs{x}/c - \log(1+\abs{x}/c))$). We refer to \cite{DBLP:conf/soda/ClarksonW15} for a discussion of $M$-estimators. When we apply this to case of $g_i(x) = x^p$ \emph{for any constant real $1\leq p<\infty$}, i.e.\ when $f$ is the $\ell_p$ norm, we obtain a polynomial time algorithm making $O\parens*{\min\braces*{n,\frac{n\log k}{\log n}}}$ queries. We note that our $\ell_p$ result generalizes the result of \cite{JiangP19} both by allowing $k$ to vary and by handling noninteger $p$. For $p=\infty$, we give a simple algorithm achieving $O(n)$ queries.

We also give lower bounds for any adaptive algorithm that match our upper bounds up to constant factors, \emph{for any constant integer $1\leq p<\infty$} (Theorem \ref{thm:lb-p}) and for $p=\infty$ (Theorem \ref{thm:lb-infty}). In fact, our lower bounds are for a weaker problem, the problem of outputting an approximation $\bfy'$ such that its distance from the true hidden vector $\bfy$ is at most $\norm*{\bfy'-\bfy}_p\leq R$, whenever the approximation radius satisfies $R\leq k^{1-\eps}n^{1/p}$ (where we think of $n^{1/p} = 1$ when $p=\infty$) for constant $\eps>0$. Thus, approximation for this problem is hard, in the sense that finding the point exactly is optimal up to constant factors, even when the approximation radius is as large as $k^{1-\eps}n^{1/p}$.

Our main algorithmic technique for obtaining Theorem \ref{thm:exact-recovery-coordinate-wise} is a judicious application of a generalization of the Fourier-based detecting matrix construction of \cite{bshouty2009optimal}. Our lower bounds are simply obtained by counting the number of lattice points in an $\ell_p$ ball.

Finally, we consider a noisy version of the above problem, where the codemaker is allowed to answer queries with any answer that is within $(1\pm\eps)\norm*{\bfy-\bfx}_p$. For this variant, we show that any algorithm must take $\Omega(\exp(\eps^2 \Theta(k^p n)))$ in Theorem \ref{thm:noisy-prob}. That is, there is no query efficient algorithm for this problem.

\section{Preliminaries}
\subsection{Notation}
\begin{dfn}[$\ell_p$ norm]
    Let $1\leq p\leq\infty$. Then, we endow $\mathbb R^n$ with the $\ell_p$ norm $\norm*{\cdot}_p$, given by
    \begin{eqn}
        \norm*{\bfx}_{p}\coloneqq \parens*{\sum_{i=1}^n \abs*{x_i}^p}^{1/p}
    \end{eqn}
    if $p<\infty$ and
    \begin{eqn}
        \norm*{\bfx}_{\infty}\coloneqq \max_{i=1}^n \abs*{x_i}
    \end{eqn}
    if $p=\infty$.
\end{dfn}

\begin{dfn}[Weight of binary vector]
    Let $a\in\{0,1\}^\nu$. Then, $\wt(a)$ is the number of $1$s in $a$.
\end{dfn}

\begin{dfn}[Even-odd decomposition]\label{def:even-odd}
    Let $h:\mathbb R\to\mathbb R$ be any function. Then, the \emph{even-odd decomposition of $h$} given by
    \begin{eqndot}
        h_\even(x) &\coloneqq \frac{h(x)+h(-x)}2 \\
        h_\odd(x) &\coloneqq \frac{h(x)-h(-x)}2
    \end{eqndot}
    It is easy to see that $h = h_\even + h_\odd$ and that $h_\even(-x) = h_\even(x)$ and $h_\odd(-x) = -h_\odd(x)$ for all $x\in\mathbb R$.
\end{dfn}

\subsection{Bshouty detecting matrix}
We very briefly review the construction of the detecting matrix of \cite{bshouty2009optimal}, as we build off of this result for our algorithms.

\begin{dfn}[Detecting matrix \cite{bshouty2009optimal}]
    A $(d_1,d_2,\dots,d_n)$-detecting matrix is a $\{0,1\}$-matrix such that for every $\bfu,\bfv\in\prod_{i=1}^n\{0,1,\dots,d_i-1\}$ with $\bfu\neq \bfv$, we have $M\bfu\neq M\bfv$.
\end{dfn}

The theorem we use is the following:
\begin{thm}[Bshouty detecting matrix, Theorem 4/Corollary 5 of \cite{bshouty2009optimal}]
    Let $1<d_1\leq d_2\leq \dots\leq d_n$ where $d_1+d_2+\dots+d_n = d$. There is a $(d_1,d_2,\dots,d_n)$-detecting matrix $M$ of size $s\times n$ where
    \begin{eqn}\label{eqn:bshouty-size-bound}
        s(\log s-4)\leq 2n\log\frac{d}{n}.
    \end{eqn}
    Furthermore, for $\bfu\in\prod_{i=1}^n\{0,1,\dots,d_i-1\}$, there is a polynomial time algorithm for recovering $\bfu$ given $M\bfu$.
\end{thm}

We will only sketch the main idea behind the construction of the matrix and the decoding algorithm, and refer the reader to \cite{bshouty2009optimal} for the proof of the bounds and the correctness.

\paragraph{Fourier representation \normalfont\cite{bshouty2009optimal}.}
We consider the Fourier basis on real-valued functions defined on the Boolean hypercube $\{-1,+1\}^\nu$, i.e.\ the basis
\begin{eqn}
    \mathcal B\coloneqq \braces*{\chi_a(x)\coloneqq \prod_{a_i=1}x_i \Bigg\vert a\in\{0,1\}^\nu}\subseteq \braces*{f:\{-1,+1\}^\nu\to\mathbb R}.
\end{eqn}
It is known that $\mathcal B$ is an orthonormal basis, and thus any $f:\{-1,+1\}^\nu\to\mathbb R$ can be uniquely represented as
\begin{eqn}
    f(x) = \sum_{a\in\{0,1\}^s}\hat f(a)\chi_a(x)
\end{eqn}
where $\hat f(a)$ is the Fourier coefficient of $\chi_a$ given by
\begin{eqn}
    \hat f(a) = \frac1{2^\nu}\sum_{x\in\{-1,+1\}^\nu} f(x)\chi_a(x).
\end{eqn}
Using the fast Fourier transform, all the coefficients $\hat f(a)$ can be found from the values of $f(x),x\in\{-1,+1\}^\nu$ and ordered according to lexicographic order of $a\in\{0,1\}^\nu$ in time $O(\nu2^\nu)$.

\paragraph{Detecting matrix construction.} The overall idea is as follows. We choose $s$ as in equation (\ref{eqn:bshouty-size-bound}) and $\nu\coloneqq \log_2 s$. Then, we view column vectors in $\mathbb R^s$ with $s = 2^\nu$ rows as enumerations of the values of functions $f:\{-1,+1\}^\nu\to\mathbb R$. That is, for $x\in\{-1,+1\}^\nu$, the $x$th row of the column vector representing $f$ is $f(x)$. We then view our detecting matrix $M\in\{0,1\}^{s\times n}$ as a family of $n$ $\{0,1\}$-valued functions defined on $\{-1,+1\}^\nu$ and $Mu$ as a linear combination of functions from this family, where the coefficients of the linear combination are specified by the unknown vector $\bfu\in\prod_{i=1}^n\{0,1,\dots,d_i-1\}$. The $n$ functions of $M$ have a special structure in the Fourier basis, so that there is an efficient iterative algorithm for recovering the coordinates of $u$ in batches from the Fourier coefficients of the function $Mu$.

We iteratively construct columns of $M$ as follows. For each $a\in\{0,1\}^\nu$, we will choose $\ell_a$ more columns to construct, so that in the end, we have $\sum_{a\in\{0,1\}^\nu}\ell_a = n$ columns.

Suppose that columns $1$ through $r$ have already been constructed. Let $a\in\{0,1\}^\nu$ and choose an integer $\ell_a$ such that
\begin{eqndot}
    d_{r+1}d_{r+2}\dots d_{r+\ell_a} &\leq 2^{\wt(a)} \\
    d_{r+1}d_{r+2}\dots d_{r+\ell_a}d_{r+\ell_a+1} &>2^{\wt(a)-1}
\end{eqndot}
We then construct $\ell_a$ more columns of $M$ so that the $i$th new function $g_{a,i}$ has Fourier coefficient of $\chi_a$ as
\begin{eqn}
    \hat g_{a,i}(a) = d_{r+1}d_{r+2}\dots d_{r+i} / 2^{\wt(a)}
\end{eqn}
and the Fourier coefficient of $\chi_b$ for any $b>a$ (in the usual ordering on the Boolean hypercube) as
\begin{eqn}
    \hat g_{a,i}(b) = 0.
\end{eqn}
The way we choose the column functions $g_{a,i}$ to have these properties is described in \cite{bshouty2009optimal}.

\paragraph{Decoding algorithm.} We now show how to efficiently decode $M\bfu$. Essentially, we will decode $\ell_a$ of the entries of $\bfu$ at a time, subtract them off, and recurse.

Note that column vector $M\bfu$ is the enumeration of the values of a linear combination $f$ of the $g_{a,i}$ functions from above, where the row corresponding to $x\in\{-1,+1\}^\nu$ is $f(x)$. Then, using the fast Fourier transform, we find all the Fourier coefficients $\hat f(z)$ for $z\in\{0,1\}^\nu$ and search for a maximal $a\in\{0,1\}^\nu$ such that $\hat f(a)\neq 0$. For such an $a$, one can prove that its Fourier coefficient in $f$ is
\begin{eqn}
    \hat f(a) = \frac1{2^{\wt(a)}}\parens*{\lambda_{r+1} + \lambda_{r+2}d_{r+1} + \lambda_{r+3}d_{r+1}d_{r+2} + \dots + \lambda_{r+\ell_a+1}d_{r+1}d_{r+2}\dots d_{r+\ell_a}}
\end{eqn}
where $r$ is the number of columns in $M$ before the columns corresponding to $a$, and $\lambda_{r+i} = u_{r+i}$ (for sake of matching the notation in \cite{bshouty2009optimal}). Since $\lambda_{r+i}\in\{0,1,\dots,d_{r+i}-1\}$ for all $i\in[\ell_a]$, we can recover all of the $\lambda_{r+i}$. Then, these coefficients can be subtracted off and we can recurse on the remaining entries of $\bfu$.

In our Theorem \ref{thm:exact-recovery-coordinate-wise}, we will modify the above algorithm to allow for non-integer values for the $\lambda_{r+i}$, as long as they are bounded and well-separated (to be made precise later).

\section{Algorithms}

We now describe our upper bounds. As a warm up, we start with algorithms for $\ell_1$, $\ell_2$, and $\ell_\infty$. These will introduce some tricks that we exploit in our coordinate-wise sums algorithm. Then, we combine these tricks along with a modification of the Bshouty detecting matrix algorithm described above to obtain Theorem \ref{thm:exact-recovery-coordinate-wise}.

\subsection{Algorithms for \texorpdfstring{$\ell_1$}{l1}, \texorpdfstring{$\ell_2$}{l2}, and \texorpdfstring{$\ell_\infty$}{l infinity}}
Our algorithms will be based around the idea of applying the Bshouty detecting matrix $M$ to the hidden vector $\bfy$. This can be most straightforwardly applied in the case of $\ell_2$, by expanding squared distances (equation (\ref{eqn:pythagorean})).

\begin{thm}[Algorithm for $\ell_2$ queries]
Let $\bfy\in\{-k,-k+1,\dots,k-1,k\}^n$ be an unknown vector, and suppose that we receive answers to $s$ queries of the form $\norm*{\bfx-\bfy}_2$. Then, there is a polynomial time algorithm that recovers $\bfy$ in $s = O\parens*{\min\braces*{n,\tfrac{n\log k}{\log n}}}$ queries.
\end{thm}
\begin{proof}
By first making the query with the $\mathbf 0$ vector, we may find the norm $\norm*{\bfy}_2$ of the unknown vector. Now suppose we query for $\norm*{\bfx-\bfy}_2$. Note then that
\begin{eqn}\label{eqn:pythagorean}
    \angle*{\bfx,\bfy} = \frac{\norm*{\bfx}_2^2 + \norm*{\bfy}_2^2 - \norm*{\bfx-\bfy}_2^2}{2}
\end{eqn}
so we can compute the inner product between $\bfx$ and $\bfy$. Thus by taking $n$ queries to be the $n$ standard basis vectors $\bfx = \bfe_i$ for $i\in[n]$, we can always recovery $\bfy$ in $n+1$ queries. To obtain $s = O\parens*{\tfrac{n\log k}{\log n}}$ queries for $k\leq n$, we can take our query vectors $\bfx$ to be the rows of the detecting matrix of Theorem 4/Corollary 5 of \cite{bshouty2009optimal} and recover $\bfy$ by using the decoding algorithm as described in the proof. We thus conclude as desired.
\end{proof}

As shown above, if we can simulate computing inner products with binary vectors in $O(1)$ queries each, then we get an $O(n)$ algorithm by querying with the standard basis vectors or $O\parens*{\tfrac{n\log k}{\log n}}$ by using \cite{bshouty2009optimal}. For $\ell_1$, we take a similar approach. This time, the way we extract the inner product is quite different from the case of $\ell_2$. This technique turns out to be much more flexible, and will allow us to generalize the result to coordinate-wise sums.

\begin{thm}[Algorithm for $\ell_1$ queries]
Let $\bfy\in\{-k,-k+1,\dots,k-1,k\}^n$ be an unknown vector, and suppose that we receive answers to $s$ queries of the form $\norm*{\bfx-\bfy}_1$. Then, there is a polynomial time algorithm that recovers $\bfy$ in $s = O\parens*{\min\braces*{n,\tfrac{n\log k}{\log n}}}$ queries.
\end{thm}
\begin{proof}
We will just show how to compute inner products in $O(1)$ queries, since the rest follows as in the $\ell_2$ case. Let $\bftau\in\{0,1\}^n$ be any binary vector and consider the sign vector $\bfsigma\in\{\pm1\}^n$ with $\sigma_i = (-1)^{\tau_i+1}$. Then for $\sigma_i\in\{\pm1\}$ and $-k\leq y_i\leq k$, we have that
\begin{eqn}
    \abs*{k\sigma_i-y_i} = \abs*{k\sigma_i-\sigma_i^2y_i} = \abs*{k-\sigma_iy_i} = k - \sigma_i y_i.
\end{eqn}
Thus,
\begin{eqn}
    \norm*{k\bfsigma-\bfy}_1 = \sum_{i=1}^n \abs*{k\sigma_i-y_i} = \sum_{i=1}^n k-\sigma_iy_i = kn - \bfsigma\cdot\bfy
\end{eqn}
so we may compute the quantity $\bfsigma\cdot\bfy = kn - \norm*{k\bfsigma-\bfy}_1$. We may then compute the desired inner product with binary vectors as $\bftau\cdot\bfy = (\bfsigma\cdot\bfy + \mathbf1_n\cdot\bfy)/2$.
\end{proof}

To conclude the section, we show an $O(n)$ algorithm for $\ell_\infty$ queries. This turns out to be optimal, as we show later.
\begin{thm}[Algorithm for $\ell_\infty$ queries]
Let $\bfy\in\{-k,-k+1,\dots,k-1,k\}^n$ be an unknown vector, and suppose that we receive answers to $s$ queries of the form $\norm*{\bfx-\bfy}_\infty$. Then, there is a polynomial time algorithm that recovers $\bfy$ in $s = O(n)$ queries.
\end{thm}
\begin{proof}
For each $i\in[n]$, we make the query $q_i^+ = \norm*{k\bfe_i-\bfy}_\infty$ and $q_i^- = \norm*{-k\bfe_i-\bfy}_\infty$. Note that $y_i = 0$ if and only if these two are both equal to $k$. If $y_i>0$, then $q_i^- = k+y_i > k$ and if $y_i<0$, then $q_i^+ = k-y_i > k$. Thus, with these two queries, we can determine $y_i$. Thus, we recover $\bfy$ in $O(n)$ queries.
\end{proof}

\subsection{Algorithm for separable distances}
In the previous section, we obtained polynomial time algorithms with tight query complexity for $\ell_1$ and $\ell_2$ by simulating inner product computations between $\bfy$ and binary vectors. We now generalize these ideas to an algorithm for any query given by sums along the coordinates. This in particular includes all $\ell_p$ norms, even for $p$ not an integer.

\begin{thm}[Algorithm for separable distances]\label{thm:exact-recovery-coordinate-wise}
Let $\bfy\in\{-k,-k+1,\dots,k-1,k\}^n$ be an unknown vector, and suppose that we receive answers to $s$ queries of the form $f(\bfy-\bfx)$, where $f(\bfx) = \sum_{i=1}^n g_i(\abs*{x_i})$. For each $i\in[n]$, define the function $h_i(x) = g_i(k-x)$ and consider the even-odd decomposition $h_i = (h_i)_\even + (h_i)_\odd$ (see Definition \ref{def:even-odd}). Also consider the following quantities:
\begin{eqn}\label{eqn:def-quantities}
    M_i^{\min} &\coloneqq \min_{x\in\{-k,-k+1,\dots,k-1,k\}}(h_i)_\odd(x) \\
    M_i^{\max} &\coloneqq \max_{x\in\{-k,-k+1,\dots,k-1,k\}}(h_i)_\odd(x) \\
    \Delta_i &\coloneqq \min_{\substack{x_1,x_2\in\{-k,-k+1,\dots,k-1,k\} \\ x_1\neq x_2}} \abs*{(h_i)_\odd(x_1) - (h_i)_\odd(x_2)}\\
    \Delta &\coloneqq \min_{i=1}^n\Delta_i \\
    d_i &\coloneqq \ceil*{\frac{M_i^{\max} - M_i^{\min}}{\Delta}}+1
\end{eqn}
If $\Delta > 0$, then there is a polynomial time algorithm that recovers $\bfy$ with $s = O\parens*{\min\braces*{n,\tfrac{\log \prod_{i=1}^n d_i}{\log n}}}$ queries.
\end{thm}
\begin{proof}
Let $\bfh_\even$ and $\bfh_\odd$ be the functions that apply $(h_i)_\even$ and $(h_i)_\odd$ on the $i$th coordinate, respectively. We will show that we can recover $\bfh_\odd(\bfy)$ in $O\parens*{\min\braces*{n,\tfrac{\log \prod_{i=1}^n d_i}{\log n}}}$ queries. Note that since $\min_{i=1}^n\Delta_i > 0$, $(h_i)_\odd$ is injective for each $i$ and thus we can recover $\bfy$ from $\bfh_\odd(\bfy)$ in polynomial time using a lookup table for the values of $(h_i)_\odd$.

\paragraph{Inner products with binary vectors.}
We first show that we can compute the inner product between $\bfh_\odd(\bfy)$ and any binary vector $\bftau\in\{0,1\}^n$. To do this, consider the sign vector $\bfsigma\in\{\pm1\}^n$ with $\sigma_i = (-1)^{\tau_i+1}$. Note that for $\sigma_i\in\{\pm1\}$ and $-k\leq y_i\leq k$, we have $\abs*{k\sigma_i-y_i} = \abs*{k-\sigma_iy_i} = k - \sigma_iy_i$. Then, by querying vectors of the form $\bfx = k\bfsigma$, we obtain
\begin{eqn}
    f(k\bfsigma-\bfy) = \sum_{i=1}^n g_i(k-\sigma_iy_i) = \sum_{i=1}^n h_i(\sigma_iy_i).
\end{eqn}
Then using the even/oddness of $(h_i)_\even$/$(h_i)_\odd$, we have
\begin{eqn}
    \sum_{i=1}^n h_i(\sigma_iy_i) = \parens*{\sum_{i=1}^n (h_i)_\even(y_i)} + \parens*{\sum_{i=1}^n\sigma_i(h_i)_\odd(y_i)} = \mathbf{1}_n\cdot \bfh_\even(\bfy) + \bfsigma\cdot\bfh_\odd(\bfy).
\end{eqn}
Note also that by querying for $k\mathbf{1}_n$ and $-k\mathbf{1}_n$, we also obtain
\begin{eqndot}
    \frac{f(k\mathbf{1}_n-\bfy) + f(-k\mathbf{1}_n-\bfy)}2 &= \sum_{i=1}^n (h_i)_\even(y_i) = \mathbf{1}_n\cdot \bfh_\even(\bfy) \\
    \frac{f(k\mathbf{1}_n-\bfy) - f(-k\mathbf{1}_n-\bfy)}2 &= \sum_{i=1}^n (h_i)_\odd(y_i) = \mathbf{1}_n\cdot \bfh_\odd(\bfy)
\end{eqndot}
Using these, we may compute $\bftau\cdot\bfh_\odd(\bfy) = \tfrac12(\bfsigma + \mathbf{1}_n)\cdot \bfh_\odd(\bfy)$ and thus we are able to compute dot products of arbitrary binary vectors with $\bfh_\odd(\bfy)$. At this point, we can obtain $O(n)$ queries just by taking the binary vectors to be the standard basis vectors, so we focus on obtaining an algorithm making at most $O\parens*{\tfrac{\log \prod_{i=1}^n d_i}{\log n}}$ queries.

\paragraph{Modification of the Bshouty detecting matrix decoding \normalfont\cite{bshouty2009optimal}.}
Recall the detecting matrix of \cite{bshouty2009optimal} for integer vectors in $\prod_{i=1}^n\{0,1,\dots,d_i-1\}$ for $d_i\in\mathbb N$ for $i\in[n]$. If $\bfh_\odd(\bfy)$ took integer values, then we could just directly use this theorem to conclude with the desired query complexity. However, this is not true of $\bfh_\odd(\bfy)$, and so we need to show how to modify the \cite{bshouty2009optimal} construction to handle our setting.

We first shift and scale our vector $\bfh_\odd(\bfy)$. Let $\bfM^{\min}$ be the vector with $M_i^{\min}$ in the $i$th coordinate. Note that we can easily compute $\bftau\cdot\bfM^{\min}$. Thus, we are able to compute dot products of arbitrary binary vectors with the vector  $(\bfh_\odd(\bfy)-\bfM^{\min})$. By dividing by $\Delta$, we have dot products of arbitrary binary vectors with$\tfrac1\Delta\parens*{\bfh_\odd(\bfy)-\bfM^{\min}}$. We now define this as
\begin{eqn}
    \varphi_i(y) &\coloneqq \frac1\Delta\parens*{(h_i)_\odd(y)-M_i^{\min}} \\
    \bfphi(\bfy) &\coloneqq \frac1\Delta\parens*{\bfh_\odd(\bfy)-\bfM^{\min}}
\end{eqn}
Note then that $0\leq \varphi_i \leq d_i-1$ (see equation (\ref{eqn:def-quantities})) and that $y_1\neq y_2\implies \abs*{\varphi(y_1)-\varphi(y_2)}\geq 1$.

Now consider the detecting matrix construction of Theorem 4 in \cite{bshouty2009optimal}. Recall that we may extract the Fourier coefficient of $\chi_a$ for some maximal $a$ in our unknown vector $\bfphi(\bfy)$ viewed as a function, which gives us
\begin{eqn}\label{eqn:fourier-coefficient}
    \lambda_{r+1}+\lambda_{r+2}d_{r+1}+\lambda_{r+3}d_{r+1}d_{r+2} + \dots + \lambda_{r+\ell_a+1}d_{r+1}d_{r+2}\dots d_{r+\ell_a}
\end{eqn}
which in our case we set $\lambda_j = \varphi_j(y_j)$. Now let $\mathcal X\coloneqq\prod_{j=r+1}^{r+\ell_a+1}\varphi_j(\{-k,-k+1,\dots,k-1,k\})$ be the image of our original points in a subset of $\ell_a+1$ coordinates starting at $r+1$ under the corresponding $\varphi_j$. Consider the function $\psi:\mathcal X\to \mathbb R^+$ defined via
\begin{eqn}
    \psi(\bfz) = \sum_{i=0}^{\ell_a} z_{i+1}\prod_{j=1}^{i}d_{r+j}.
\end{eqn}
It is easy to see that when we endow $\mathcal X$ with the lexicographical ordering, then $\psi$ is increasing. Thus, given the Fourier coefficient as in equation (\ref{eqn:fourier-coefficient}), we can do binary search on the at most $k^n$ values in $\mathcal X$ to extract the values $\lambda_{r+i}$ in time $O(n\log k)$. Given this step of recovering $\ell_a$ of the coordinates, we can proceed as in the rest of \cite{bshouty2009optimal} by subtracting these coordinates of the unknown vector and recursing. Hence, we conclude that we may recover $\bfphi(\bfy)$ efficiently and thus $\bfh(\bfy)$, as claimed.
\end{proof}

\subsubsection{Applications}
As a corollary of the above result, we obtain an algorithm for recovering $\bfy\in\{-k,-k+1,\dots,k-1,k\}^n$ from $s = O\parens*{\min\braces*{n,\tfrac{n\log k}{\log n}}}$ distance queries in $\ell_p$ and a variety of $M$-estimators. We also give a weaker bound of $O(\min\{n,nk/\log n\})$ for the smooth max. The proofs are deferred to Appendix \ref{section:applications}.

\begin{cor}[Algorithm for $\ell_p$ queries]\label{cor:lp}
Let $\bfy\in\{-k,-k+1,\dots,k-1,k\}^n$ be an unknown vector, and suppose that we receive answers to $s$ queries of the form $\norm*{\bfx-\bfy}_p$ for $p$ a constant. Then, there is a polynomial time algorithm that recovers $\bfy$ in $s = O\parens*{\min\braces*{n,\tfrac{n\log k}{\log n}}}$ queries.
\end{cor}

\begin{cor}[Algorithm for $M$-estimator loss queries]\label{cor:m}
Let $\bfy\in\{-k,-k+1,\dots,k-1,k\}^n$ be an unknown vector, and suppose that we receive answers to $s$ queries of the form $\norm*{\bfx-\bfy}_M = \sum_{i=1}^n M(\abs{x_i-y_i})$ for any one of the following choices of $M$:
\begin{itemize}
    \item \textbf{$\ell_1$-$\ell_2$ loss}
    \begin{eqn}
        M(x) = 2\parens*{\sqrt{1+\frac{\abs{x}^2}{2}}-1}
    \end{eqn}
    \item \textbf{Huber loss}
    \begin{eqn}
        M(x) = \begin{cases}
            \abs{x}^2/2\tau & \text{if $\abs{x}\leq\tau$} \\
            \abs{x}-\tau/2 & \text{otherwise}
        \end{cases}
    \end{eqn}
    \item \textbf{Fair estimator loss}
    \begin{eqn}
        M(x) = c^2\parens*{\frac{\abs{x}}{c}-\log\parens*{1+\frac{\abs{x}}{c}}}
    \end{eqn}
\end{itemize}
Then, there is a polynomial time algorithm that recovers $\bfy$ in $s = O\parens*{\min\braces*{n,\tfrac{n\log k}{\log n}}}$ queries.
\end{cor}

\begin{cor}[Algorithm for smooth max queries]\label{cor:sm}
Let $\bfy\in\{-k,-k+1,\dots,k-1,k\}^n$ be an unknown vector, and suppose that we receive answers to $s$ smooth max queries. Then, there is a polynomial time algorithm that recovers $\bfy$ in $s = O\parens*{\min\braces*{n,\tfrac{nk}{\log n}}}$ queries.
\end{cor}

\section{Lower Bounds}
In this section, we compliment our $\ell_p$ algorithms with matching lower bounds, for integer $p$. Our lower bounds work even for the problem of approximating the hidden vector and for adaptive randomized algorithms with constant success probability, improving upon the nonadaptive lower and exact lower bound of \cite{JiangP19}.

\begin{thm}[Lower bound for integer $\ell_p$]\label{thm:lb-p}
Let $1\leq p<\infty$ be a constant integer and let $R\in (0,kn^{1/p}]$ be an approximation radius. Suppose there exists an algorithm $\mathcal A$ such that for all unknown vectors $\bfy\in\{-k,-k+1,\dots,k-1,k\}^n$, $\mathcal A$ outputs a vector $\bfy'\in\{-k,-k+1,\dots,k-1,k\}^n$ such that
\begin{eqn}
    \norm*{\bfy'-\bfy}_p\leq R
\end{eqn}
in $s$ possibly adaptive $\ell_p$ queries with probability at least $2/3$ over the algorithm's random coin tosses. Then
\begin{eqn}
    s = \Omega\parens*{\frac{n\log(kn^{1/p}/R)}{\log k + \log n}}.
\end{eqn}
In particular, if $R\leq k^{1-\eps}n^{1/p}$ for some constant $\eps>0$, then
\begin{eqn}
    s = \Omega\parens*{\frac{n\log k}{\log k + \log n}},
\end{eqn}
which is $\Omega\parens*{\frac{n\log k}{\log n}}$ if $k < n$ and $\Omega(n)$ if $k\geq n$.
\end{thm}
\begin{proof}
By Yao's minimax principle \cite{yao1977probabilistic}, it suffices to show the lower bound for all deterministic algorithms $\mathcal A$ that correctly approximates a uniformly random $\bfy\in\{-k,-k+1,\dots,k-1,k\}^n$ with probability at least $2/3$.

Note that each query $\norm*{\bfx-\bfy}_p^p$ results in a nonnegative integer that is at most $(2k)^pn$. Thus, there are at most $((2k)^pn + 1)^s$ possible sequences of answers. Now let $Q$ be the set of all sequence of answers that $\mathcal A$ can observe, and for each sequence of answers $\bfq\in Q$, let $S_\bfq$ denote the set of vectors $\bfy\in\{-k,-k+1,\dots,k-1,k\}^n$ such that the deterministic algorithm $\mathcal A$ observes $\bfq$ on input $\bfy$. Then, $S_\bfq$ partitions the unknown vectors $\bfy$ into $\abs*{Q}$ disjoint sets. Then, the probability that $\abs*{S_\bfq}$ has size at most $\tfrac1{100}\tfrac{(2k+1)^n}{\abs*{Q}}$ is
\begin{eqn}
    \Pr_\bfy\parens*{\abs*{S_\bfq}\leq\frac1{100}\frac{(2k+1)^n}{\abs*{Q}}} &= \sum_{\substack{\bfq\in Q \\ \abs*{S_\bfq}\leq \frac1{100}\frac{(2k+1)^n}{\abs*{Q}}}}\Pr_\bfy\parens*{\text{$\mathcal A$ queries the sequence $\bfq$}} \\
    &= \sum_{\substack{\bfq\in Q \\ \abs*{S_\bfq}\leq \frac1{100}\frac{(2k+1)^n}{\abs*{Q}}}}\frac{\abs*{S_\bfq}}{(2k+1)^n} \\
    &\leq \sum_{\substack{\bfq\in Q \\ \abs*{S_\bfq}\leq \frac1{100}\frac{(2k+1)^n}{\abs*{Q}}}}\frac1{100}\frac{(2k+1)^n}{\abs*{Q}}\frac1{(2k+1)^n} \\
    &\leq \sum_{\bfq\in Q}\frac1{100\abs*{Q}} = \frac1{100}.
\end{eqn}
Thus with probability at least $99/100$, $\abs*{S_\bfq}$ has size at least $\tfrac1{100}\tfrac{(2k+1)^n}{\abs*{Q}}$.

Note that by \cite{wang2005volumes}, the volume of a unit $\ell_p$ ball is $2^n\Gamma(1+1/p)^n/\Gamma(1+n/p)$, so the volume of a ball of radius $R$ in $\ell_p$ is
\begin{eqn}
    V \coloneqq R^n 2^n\frac{\Gamma(1+1/p)^n}{\Gamma(1+n/p)}= \parens*{\Theta\parens*{\frac{R}{n^{1/p}}}}^n.
\end{eqn}
Now suppose that $\bfq$ is a sequence of queries such that $\abs*{S_\bfq}>2V$ and let $\bfz$ be the output of the deterministic algorithm $\mathcal A$ on the sequence of queries $\bfq$. Then, at most $V$ of the points in $S$ can be in the $\ell_p$ ball of radius $R$ centered at $\bfz$. Thus, with probability at least $1/2$ over the random hidden vector $\bfy$, we output a point $\bfz$ such that $\norm*{\bfz-\bfy}_p\geq R$. Thus, if
\begin{eqn}
    \frac1{100}\frac{(2k+1)^n}{\abs*{Q}} > 2V,
\end{eqn}
then our probability of success is at most $1/2 + 1/100$ and thus we do not have a correct algorithm. Thus, it must be that
\begin{eqn}
    \frac1{100}\frac{(2k+1)^n}{\abs*{Q}}\leq 2V\implies \frac{(2k+1)^n}{200V}\leq \abs*{Q}\leq ((2k)^pn+1)^s.
\end{eqn}
Rearranging, we have that
\begin{eqn}
    s\geq\frac{\log\frac{(2k+1)^n}{200V}}{\log((2k)^pn + 1)} = \Omega\parens*{\frac{n\log(kn^{1/p}/R)}{\log k + \log n}},
\end{eqn}
as claimed.
\end{proof}

For $p = \infty$, we have a lower bound of $\Omega(n)$ regardless of $k$.
\begin{thm}[Lower bound for $\ell_\infty$]\label{thm:lb-infty}
Let $R\in (0,k]$ be an approximation radius. Suppose there exists an algorithm $\mathcal A$ such that for all unknown vectors $\bfy\in\{-k,-k+1,\dots,k-1,k\}^n$, $\mathcal A$ outputs a vector $\bfy'\in\{-k,-k+1,\dots,k-1,k\}^n$ such that
\begin{eqn}
    \norm*{\bfy'-\bfy}_\infty\leq R
\end{eqn}
in $s$ possibly adaptive $\ell_\infty$ queries with probability at least $2/3$ over the algorithm's random coin tosses. Then
\begin{eqn}
    s = \Omega\parens*{\frac{n\log(k/R)}{\log k}}.
\end{eqn}
In particular, if $R\leq k^{1-\eps}$ for a constant $\eps>0$, then $s = \Omega(n)$.
\end{thm}
\begin{proof}
By the same argument as the finite $\ell_p$ case, we use Yao's minimax principle to pass the argument to a lower bound for all deterministic algorithms $\mathcal A$ on uniformly random inputs $\bfy$ succeeding with probability at least $2/3$. Furthermore, by the same partition argument as before, we have that $\abs*{S_\bfq}$ is at least $\tfrac1{100}\tfrac{(2k+1)^n}{\abs*{Q}}$ with probability at least $99/100$.

The volume of an $\ell_\infty$ ball of radius $R$ is $(2R)^n$, so as before, we must have
\begin{eqn}
    \frac1{100}\frac{(2k+1)^n}{\abs*{Q}}\leq 2(2R)^n.
\end{eqn}
When $p=\infty$, there are only $(2k+1)^s$ possible sequences of answers, so we instead have the bound
\begin{eqn}
    \frac{(2k+1)^n}{(2R)^n}\leq 200(2k+1)^s
\end{eqn}
By rearranging, we obtain the bound $s = \Omega\parens*{\frac{n\log(k/R)}{\log k}}$ as desired.
\end{proof}

\subsection{Lower bound for the noisy problem}
Finally, we show that in the noisy version of the problem, i.e.\ the setting where the codemaker is allowed to answer the queries $\bfx$ with any $q = (1\pm\eps)\norm*{\bfy-\bfx}_p$, there is no good algorithm.
\begin{thm}[Lower bound for the noisy problem]\label{thm:noisy-prob}
Let $1\leq p<\infty$ be constant and let $0<R<kn^{1/p}$ be an approximation radius. Suppose there exists an algorithm $\mathcal A$ such that for all unknown vectors $\bfy\in\{-k,-k+1,\dots,k-1,k\}^n$, $\mathcal A$ outputs a vector $\bfy'\in\{-k,-k+1,\dots,k-1,k\}^n$ such that
\begin{eqn}
    \norm*{\bfy'-\bfy}_p\leq R
\end{eqn}
in $s$ possibly adaptive $(1\pm\eps)$-noisy $\ell_p$ queries, i.e.\ answers with adversarially chosen $q_\bfx = (1\pm\eps)\norm*{\bfy-\bfx}_p$, with probability at least $2/3$ over the algorithm's random coin tosses. Then
\begin{eqn}
    s = \Omega\parens*{\exp\parens*{\eps^2\Theta(k^p n)}}.
\end{eqn}
\end{thm}
\begin{proof}
By Yao's minimax principle, we can take the algorithm to be deterministic by taking our hidden vector $\bfy$ to be drawn uniformly from $\{-k,-k+1,\dots,k-1,k\}^n$. Now fix any query $\bfx\in\{-k,-k+1,\dots,k-1,k\}^n$ and let $\mu = \bfE_\bfz\parens*{\norm*{\bfx-\bfz}_p^p} = \Theta(k^pn)$. Then by Chernoff bounds,
\begin{eqn}
    \Pr_{\bfy}\parens*{\abs*{\norm*{\bfx-\bfy}_p^p-\mu}\geq \eps\mu}\leq 2\exp\parens*{-\eps^2\mu}.
\end{eqn}
Thus, if the number of queries $s$ is less than $\exp\parens*{\eps^2\Theta(k^p n)}/200$, then by the union bound over the $s$ queries, with probability at least $99/100$ over the choice of $\bfy$, the codemaker can just return $\bfE_\bfz\parens*{\norm*{\bfx-\bfz}_p^p}$ for any query $\bfx$. Thus, the deterministic codebreaker algorithm sees the same sequence of answers with probability at least $99/100$ and so the algorithm cannot be correct. Hence, we conclude that $s = \Omega(\exp\parens*{\eps^2\Theta(k^p n)})$.
\end{proof}

\section{Acknowledgements}
We thank Flavio Chierichetti and Ravi Kumar for helpful discussions, as well
as the anonymous reviewers for helpful feedback. We also thank Zilin Jiang and Nikita Polianskii for pointing out connections to their work in \cite{JiangP19}.

\bibliographystyle{alpha}
\bibliography{citations}

\appendix

\section{Applications of Theorem \ref{thm:exact-recovery-coordinate-wise}}\label{section:applications}
In this section, we include computations that we need to apply our separable distances algorithm in Theorem \ref{thm:exact-recovery-coordinate-wise} to $\ell_p$ queries, $M$-estimator loss queries, and smooth max queries.

\paragraph{$\ell_p$ queries.}

\begin{proof}[Proof of Corollary \ref{cor:lp}]
We are in the setting to use Theorem \ref{thm:exact-recovery-coordinate-wise}, with $g_i(x) = x^p$ for all $i\in[n]$ and $h_\odd(x) = \tfrac12((k-x)^p - (k+x)^p)$. Recall that we have efficient algorithms with the desired guarantees when $p\in\{1,2\}$ so we dismiss these cases. In the remaining range of $p$, we just need to compute $\prod_{i=1}^n d_i$.

Note that
\begin{eqn}
    h_\odd'(x) = -\frac{p}2\parens*{(k-x)^{p-1} + (k+x)^{p-1}} < 0
\end{eqn}
on $x\in [-k,k]$ so $h_\odd$ is decreasing on this interval. Then, $(2k)^p/2 = h_\odd(-k)\geq h_\odd(x)\geq h_\odd(k) = -(2k)^p/2$. Furthermore, note that
\begin{eqn}
    h_\odd''(x) =  \frac{p(p-1)}2\parens*{(k-x)^{p-2}-(k+x)^{p-2}}.
\end{eqn}
If $p > 2$, then this is negative on $x > 0$, so $\abs*{h_\odd'(x)}$ is smallest at $x = 0$ and thus $\abs*{h_\odd'(x)}\geq \abs*{h_\odd'(0)} = pk^{p-1}$ for all $x$. If $1 < p < 2$, then this is positive on $x\geq 0$, so $\abs*{h_\odd'(x)}$ is smallest at $x = k$ and thus $\abs*{h_\odd'(x)}\geq \abs*{h_\odd'(k)} = (p/2)(2k)^{p-1}$ for all $x$. In either case, we have that $\Delta = \Omega(pk^{p-1})$ and the range is $M^{\max}-M^{\min} = O(k^p)$ and thus $d_i =  O((M^{\max}-M^{\min})/\Delta) = O(k)$. Thus, the query complexity is
\begin{eqn}
    O\parens*{\min\braces*{n, \frac{\log\prod_{i=1}^n d_i}{\log n}}} = O\parens*{\min\braces*{n, \frac{n\log k}{\log n}}}
\end{eqn}
as desired.
\end{proof}

\paragraph{$M$-estimator loss queries.}

\begin{proof}[Proof of Corollary \ref{cor:m}]
For the $\ell_1$-$\ell_2$ loss, we have that
\begin{eqn}
    h_\odd(x) = \sqrt{1+\frac{\abs{k-x}^2}{2}} - \sqrt{1+\frac{\abs{k+x}^2}{2}}
\end{eqn}
so
\begin{eqn}
    h_\odd'(x) = \frac12\parens*{\frac{x-k}{\sqrt{1+\abs{k-x}^2/2}} - \frac{x+k}{\sqrt{1+\abs{k+x}^2/2}}} < 0
\end{eqn}
on $x\in[-k,k]$ so $h_\odd$ is decreasing on this interval. Then, $\sqrt{1+2k^2} = h_\odd(-k)\geq h_\odd(k) = -\sqrt{1+2k^2}$. Furthermore, one can check that the second derivative is positive on $x>0$, so $\abs{h_\odd'(x)}$ is smallest at $x=k$ and thus $\abs{h_\odd'(x)}\geq \abs{h_\odd'(k)} = k/\sqrt{1+2k^2}$ for all $x$. Thus, $\Delta = \Omega(1)$ and the range is $M^{\max}-M^{\min} = O(k)$ and thus $d_i = O(k)$.

For the Huber loss, we assume that $k>\tau$, since otherwise the Huber loss is just the $\ell_2$ loss. We also restrict our attention to $x\in[0,k]$ by symmetry. Then, we have that
\begin{eqn}
    h_\odd(x) =
    \begin{cases}
        \frac12\parens*{\abs{k-x}-\abs{k+x}} = -x & \text{if $x\in[0,k-\tau]$} \\
        \frac12\parens*{\frac{\abs{k-x}^2}{2\tau}-\abs{k+x}+\frac\tau2} = \frac12\parens*{\frac{\abs{k-x}^2}{2\tau}-(k+x)+\frac\tau2} & \text{if $x\in(k-\tau,k]$}
    \end{cases}
\end{eqn}
so
\begin{eqn}
    h_\odd'(x) =
    \begin{cases}
        -1 & \text{if $x\in[0,k-\tau]$} \\
        \frac{x-(k+\tau)}{2\tau} & \text{if $x\in(k-\tau,k]$}
    \end{cases} < 0.
\end{eqn}
The range is thus $\abs{h_\odd(k)} = 2k-\tau/2 = O(k)$ and $\Delta = \Omega(1)$ so $d_i = O(k)$.

For the Fair estimator loss, restricting our attention again to $x\in[0,k]$, we have that
\begin{eqn}
    h_\odd(x) = \frac12c^2\parens*{\frac{\abs{k-x}}{c}-\frac{\abs{k+x}}{c} - \log(1+\abs{k-x}/c)) + \log(1+\abs{k+x}/c))} = -cx + \frac{c^2}2\log\frac{c+k+x}{c+k-x}
\end{eqn}
so
\begin{eqn}
    h_\odd'(x) = -c+\frac{c^2}2\frac{c+k-x}{c+k+x}\frac{2(c+k)}{(c+k-x)^2} = -c + \frac{c^2(c+k)}{(c+k)^2-x^2} \leq -c + \frac{c(c+k)}{c+2k} < 0.
\end{eqn}
The range is thus $\abs{h_\odd(k)} = O(k)$ and $\Delta = \Omega(1)$ so $d_i = O(k)$.
\end{proof}

\paragraph{Smooth max queries.}

\begin{proof}[Proof of Corollary \ref{cor:sm}]
For the smooth max, restricting our attention to $x\in[0,k]$, we have that
\begin{eqn}
    h_\odd(x) = \frac12\parens*{\exp(\abs{k-x}) + \exp(\abs{k+x})} = e^k\frac{e^x-e^{-x}}2
\end{eqn}
which has range $O(e^{2k})$ and $\Delta = \Omega(1)$.
\end{proof}

\end{document}